\documentclass[11pt]{article}

\usepackage[margin=1in]{geometry}
\usepackage{graphicx}
\usepackage[pdftex,colorlinks=true,linkcolor=blue,citecolor=blue,urlcolor=black]{hyperref}
\usepackage{amsmath, amsthm, amssymb}
\usepackage{subfigure}
\usepackage{comment}
\usepackage{url}
\usepackage{pdflscape}

\newcommand{\C}{\mathbb{C}}

\newcommand{\E}{\mathbb{E}}

\newcommand{\ket}[1]{| #1 \rangle}
\newcommand{\bra}[1]{\langle #1|}

\DeclareMathOperator{\poly}{poly}

\DeclareMathOperator{\tr}{tr}

\newcommand{\be}{\begin{equation}}
\newcommand{\ee}{\end{equation}}
\newcommand{\bea}{\begin{eqnarray}}
\newcommand{\eea}{\end{eqnarray}}
\newcommand{\bes}{\begin{equation*}}
\newcommand{\ees}{\end{equation*}}
\newcommand{\beas}{\begin{eqnarray*}}
\newcommand{\eeas}{\end{eqnarray*}}

\def\clap#1{\hbox to 0pt{\hss#1\hss}}

\makeatletter
\newtheorem*{rep@theorem}{\rep@title}
\newcommand{\newreptheorem}[2]{%
\newenvironment{rep#1}[1]{%
 \def\rep@title{#2 \ref{##1} (restated)}%
 \begin{rep@theorem}}%
 {\end{rep@theorem}}}
\makeatother

\newtheorem{thm}{Theorem}
\newtheorem*{thm*}{Theorem}

\newtheorem{lem}[thm]{Lemma}
\newtheorem*{lem*}{Lemma}
\newtheorem{prop}[thm]{Proposition}
\newtheorem{defn}[thm]{Definition}

\newtheorem{obs}[thm]{Observation}
\newtheorem{fact}[thm]{Fact}

\newreptheorem{thm}{Theorem}
\newreptheorem{lem}{Lemma}

\begin{document}

\title{Weak multiplicativity for random quantum channels}

\author{Ashley Montanaro\footnote{Centre for Quantum Information and Quantum Foundations, Department of Applied Mathematics and Theoretical Physics, University of Cambridge, UK; {\tt am994@cam.ac.uk}.}}

\maketitle

\begin{abstract}
It is known that random quantum channels exhibit significant violations of multiplicativity of maximum output $p$-norms for any $p>1$. In this work, we show that a weaker variant of multiplicativity nevertheless holds for these channels. For any constant $p>1$, given a random quantum channel $\mathcal{N}$ (i.e.\ a channel whose Stinespring representation corresponds to a random subspace $S$), we show that with high probability the maximum output $p$-norm of $\mathcal{N}^{\otimes n}$ decays exponentially with $n$. The proof is based on relaxing the maximum output $\infty$-norm of $\mathcal{N}$ to the operator norm of the partial transpose of the projector onto $S$, then calculating upper bounds on this quantity using ideas from random matrix theory.
\end{abstract}


\section{Introduction}

For many years, some of the most vexatious open problems of quantum information theory have concerned maximum output $p$-norms of quantum channels. If $\mathcal{N}$ is a quantum channel (i.e.\ completely positive, trace-preserving map), the maximum output $p$-norm of $\mathcal{N}$ is defined as
\[ \|\mathcal{N}\|_{1\rightarrow p} := \max \{ \|\mathcal{N}(\rho)\|_p,\, \rho \ge 0,\, \tr \rho = 1 \}, \]
where $\|X\|_p := (\tr |X|^p)^{1/p}$ is the Schatten $p$-norm. (The notation $\nu_p(\mathcal{N})$ is also used for $\|\mathcal{N}\|_{1\rightarrow p}$. Technically, $\|\mathcal{N}\|_{1\rightarrow p}$ is actually defined as $\sup \{ \|\mathcal{N}(X)\|_p/\|X\|_1\}$, where the supremum is taken over all non-zero Hermitian operators $X$, but it can be shown that this definition is equivalent~\cite{amosov02}.) It was a long-standing conjecture in quantum information theory~\cite{amosov00} that, for any two quantum channels $\mathcal{N}_1$, $\mathcal{N}_2$,
\[ \| \mathcal{N}_1 \otimes \mathcal{N}_2 \|_{1\rightarrow p}  \stackrel{?}{=} \| \mathcal{N}_1\|_{1\rightarrow p}  \| \mathcal{N}_2 \|_{1\rightarrow p}, \]
at least for $p$ fairly close to 1. This is equivalent to the question of additivity of {\em minimum} output R\'enyi $p$-entropies, which are defined in terms of maximum output $p$-norms as
\[ H_p^{\min}(\mathcal{N}) := \frac{1}{1-p} \log \|\mathcal{N}\|_{1\rightarrow p}^p.\]
The minimum output (von Neumann) entropy $H^{\min}(\mathcal{N})$ is obtained by taking the limit $p \rightarrow 1$~\cite{amosov00}. This case of the additivity question was of particular interest due to its connections with many other additivity problems in quantum information theory~\cite{shor04}.

All of these multiplicativity/additivity conjectures are now known to be false. First, Werner and Holevo found a counterexample to multiplicativity for $p>4.79$~\cite{werner02}. Some years later, the conjecture was falsified in the range $p>2$ by Winter~\cite{winter07}, which was swiftly extended to $1<p<2$ by Hayden~\cite{hayden07}. These works were combined as~\cite{hayden08}, which also includes the remaining case $p=2$. One can generalise the conjectures to $p<1$ (where $\|\cdot\|_p$ is of course no longer a norm), and in this setting Cubitt et al~\cite{cubitt08a} falsified additivity of minimum output R\'enyi $p$-entropies for $p \approx 0$. Finally, Hastings showed that the minimum output entropy is not additive~\cite{hastings09}. Following this, Aubrun, Szarek and Werner showed that the results of Hayden, Winter and Hastings can be obtained from Dvoretzky's theorem in the language of asymptotic geometric analysis~\cite{aubrun10,aubrun11}.

As well as the limit $p\rightarrow 1$, another important special case of the multiplicativity question is $p=\infty$, which turns out to be closely related to a number of other quantities studied in quantum information theory, as we now discuss. Any quantum channel performing a map from a $d_A$-dimensional quantum system $A$ to a $d_B$-dimensional quantum system $B$ can be written as $\mathcal{N}(\rho) = \tr_E V \rho V^\dag$ for some isometry $V: \C^{d_A} \rightarrow \C^{d_B} \otimes \C^{d_E}$ (a form known as the Stinespring dilation). The operator $M = V V^\dag$ projects onto a subspace $S \subseteq \C^{d_B} \otimes \C^{d_E}$. For our purposes, we can simply identify $\mathcal{N}$ with either $V$, $S$ or $M$.

Let $\operatorname{SEP} \subset \mathcal{B}(\C^{d_A} \otimes \C^{d_B})$ be the set of $d_A \times d_B$-dimensional separable quantum states. For any operator $M \in \mathcal{B}(\C^{d_A} \otimes \C^{d_B})$ such that $0 \le M \le I$, the quantity
\[ h_{\operatorname{SEP}}(M) := \max_{\rho \in \text{SEP}} \tr M \rho \]
is known as the support function of the separable states, evaluated at $M$. This quantity has the following connection to maximum output $p$-norms:

\begin{fact}
Let $\mathcal{N}$ be a quantum channel with corresponding isometry $V$, and set $M = V V^\dag$. Then
\[ h_{\operatorname{SEP}}(M) = \|\mathcal{N}\|_{1 \rightarrow \infty}. \]
\end{fact}

This fact can easily be proven using the Schmidt decomposition, and indeed can be generalised to arbitrary operators $0 \le M \le I$~\cite{matsumoto05} (see~\cite{harrow10} for a proof). The quantity $h_{\operatorname{SEP}}$ is crucially important in the study of multiple-prover quantum Merlin-Arthur games~\cite{kobayashi03,harrow10}, which we now briefly discuss. The complexity class $\mathsf{QMA}(2)$ is informally defined as the class of decision problems which can be solved by a polynomial-time quantum verifier (Arthur) given access to two {\em unentangled} quantum states (or ``proofs'') produced by two all-powerful but potentially malicious provers (Merlin A and Merlin B). Consider an instance of a $\mathsf{QMA}(2)$ problem for which Arthur should output ``no''. If $M$ denotes Arthur's measurement operator corresponding to a ``yes'' outcome, the maximal probability with which the two Merlins can convince Arthur to (incorrectly) output ``yes'' is precisely $h_{\operatorname{SEP}}(M)$. Therefore, projectors $M$ such that $h_{\operatorname{SEP}}(M^{\otimes n}) = h_{\operatorname{SEP}}(M)^n$ correspond to measurement operators occurring in two-prover quantum Merlin-Arthur games which obey {\em perfect parallel repetition}, i.e.\ where Arthur can simply repeat the protocol $n$ times in parallel to reduce a failure probability of $s$ to a failure probability of $s^n$. The failure of multiplicativity for $\|\mathcal{N}\|_{1 \rightarrow \infty}$ implies that such a precise form of parallel repetition cannot hold in general; however, it could still be the case that a weaker form of parallel repetition holds, where $h_{\operatorname{SEP}}(M^{\otimes n})$ necessarily decreases exponentially with $n$.

$h_{\operatorname{SEP}}$ also turns out to have many other connections to important quantities in quantum information theory and tensor optimisation (see \cite{harrow10} for some examples).


\subsection{Counterexamples to multiplicativity}

The known counterexamples to multiplicativity fall into two classes: explicit and randomised. The known explicit counterexamples are a channel of Werner and Holevo~\cite{werner02}, which acts on $d$-dimensional square matrices $\rho$ by the map
\[ \rho \mapsto \frac{1}{d-1} \left((\tr \rho) I - \rho^T \right), \]
and the channel whose corresponding subspace in the Stinespring form is the antisymmetric subspace of $\C^d \otimes \C^d$~\cite{grudka10}. This latter channel violates multiplicativity for all $p>2$, and is of particular interest due both to its simplicity and the fact that it displays a very strong violation of multiplicativity when $p=\infty$. Indeed, if $P_{\operatorname{anti}}$ denotes the projector onto the antisymmetric subspace, it holds that
\[ h_{\operatorname{SEP}}(P_{\operatorname{anti}}) = \frac{1}{2},\text{ but } h_{\operatorname{SEP}}(P_{\operatorname{anti}}^{\otimes 2}) \ge \frac{1}{2}\left(1 - \frac{1}{d} \right). \]
The counterexamples of Hayden and Winter~\cite{winter07,hayden07,hayden08} are random constructions. In particular, the construction used to falsify $p$-norm multiplicativity for all $p>1$ is to choose the first channel $\mathcal{N}$'s corresponding subspace $S \subset \C^d \otimes \C^d$ at random from the set of all subspaces of dimension $r = O(d^{1+1/p})$ (i.e.\ according to Haar measure on the unitary group), and to take $\bar{\mathcal{N}}$ as the second channel.

In the case $p=\infty$, the violation of multiplicativity displayed by this construction is also near-maximal. In other words,
\[ \| \mathcal{N} \otimes \bar{\mathcal{N}} \|_{1\rightarrow \infty} \approx \| \mathcal{N} \|_{1\rightarrow \infty}. \]
(One can show that, for any channels $\mathcal{N}_1$ and $\mathcal{N}_2$, $\| \mathcal{N}_1 \otimes \mathcal{N}_2 \|_{1\rightarrow p} \le \|\mathcal{N}_1\|_{1\rightarrow p}$~\cite{amosov00}.) While the example of the antisymmetric subspace implies that there exists a channel $\mathcal{N}$ such that one can achieve a much larger output $p$-norm by using an entangled state as input to $\mathcal{N}^{\otimes 2}$ than is possible using only product states, it leaves open the question of the general behaviour of $\| \mathcal{N}^{\otimes n} \|_{1\rightarrow p}$ for larger $n$. To the author's knowledge, two extreme situations are still possibilities: on the one hand, it might hold that
\[ \| \mathcal{N}^{\otimes n} \|_{1\rightarrow p} \stackrel{?}{\le} \| \mathcal{N} \|_{1\rightarrow p}^{n/2} \]
for all $\mathcal{N}$; alternatively, there might be no universal constant $\alpha$ such that, for all channels $\mathcal{N}$,
\[ \| \mathcal{N}^{\otimes n} \|_{1\rightarrow p} \le \| \mathcal{N} \|_{1\rightarrow p}^{\alpha n}. \]
The former possibility would imply that the largest possible violation of multiplicativity is quite mild, and in the case $p=\infty$ that a form of parallel repetition holds for two-prover quantum Merlin-Arthur games; the latter would mean that severe violations are possible and parallel repetition fails.

Interestingly, for the antisymmetric subspace it turns out that the former possibility is closer to the truth. In a beautiful and technical recent work, Christandl, Schuch and Winter~\cite{christandl09,christandl10} have shown that there is a constant $C>0$ such that
\[ h_{\operatorname{SEP}}(P_{\operatorname{anti}}^{\otimes n}) \le 2^{-Cn}, \]
implying that a weak variant of multiplicativity does indeed hold for this channel.


\subsection{New results}

The main result of this work is that, even though random quantum channels do not obey multiplicativity, their violations of multiplicativity are in some sense also very weak.

\begin{defn}
A quantum channel $\mathcal{N}$ obeys weak $p$-norm multiplicativity with exponent $\alpha$ if, for all $n \ge 1$,
\[ \| \mathcal{N}^{\otimes n} \|_{1\rightarrow p} \le \| \mathcal{N} \|_{1\rightarrow p}^{\alpha n}. \]
\end{defn}

Observe that, for any $p > 1$,
\[ \| \mathcal{N}^{\otimes n} \|_{1\rightarrow p} \le \| \mathcal{N}^{\otimes n} \|_{1\rightarrow \infty}^{1-1/p}, \]
which follows from the (matrix) H\"older inequality $\|X\|_p^p \le \|X\|_1\|X\|_\infty^{p-1}$ for any $X$. So, if $\mathcal{N}$ obeys weak $\infty$-norm multiplicativity with exponent $\alpha$, $\mathcal{N}$ also obeys weak $p$-norm multiplicativity for any $p>1$, with exponent $\alpha(1-1/p)$.

Our main result can be summarised informally as follows (a more technical version is given as Theorem \ref{thm:toplevel} below).

\begin{thm}
\label{thm:intro}
Let $\mathcal{N}$ be a quantum channel whose corresponding subspace is a random dimension $r$ subspace of $\C^{d_A} \otimes \C^{d_B}$, set $m := \min\{r,d_A,d_B\}$, and assume $m \ge 2 (\log_2 \max\{d_A,d_B\})^{3/2}$ and $r = o(d_A d_B)$. Then the probability that $\mathcal{N}$ does {\em not} obey weak $\infty$-norm multiplicativity with exponent $1/2-o(1)$ (if $r \ge d_B/d_A$), or $1-o(1)$ (if $r \le d_B/d_A$) is exponentially small in $m$.
\end{thm}

Note that we have switched notation slightly from $\C^{d_B} \otimes \C^{d_E}$ to $\C^{d_A} \otimes \C^{d_B}$, and will use this notation henceforth. In Theorem \ref{thm:intro}, and throughout the rest of the paper, our notion of a random dimension $r$ subspace $S \subseteq \C^{d_A} \otimes \C^{d_B}$ is as follows: the projector $M$ onto $S$ is formed by starting with the projector $M_0$ onto an arbitrary fixed dimension $r$ subspace $S_0 \subseteq \C^{d_A} \otimes \C^{d_B}$, and taking $M = U M_0 U^{\dag}$, where $U$ is a random unitary operator (i.e.\ $U$ is picked according to the Haar measure on the unitary group $U(d_A d_B)$). We usually think of $r$, $d_A$ and $d_B$ as all growing, but not necessarily at the same rate.

Theorem~\ref{thm:intro} implies that, for random channels $\mathcal{N}$ satisfying some mild dimensionality constraints, with high probability
\[ \| \mathcal{N}^{\otimes n} \|_{1\rightarrow p} \le \| \mathcal{N} \|_{1\rightarrow p}^{(1/2-o(1))(1-1/p) n}, \]
so random channels obey weak $p$-norm multiplicativity with exponent $(1/2-o(1))(1-1/p)$. Note that the results of Hayden and Winter imply that when $r \approx d_A = d_B$, with high probability
\[ \| \mathcal{N} \otimes \bar{\mathcal{N}} \|_{1\rightarrow \infty} \approx \| \mathcal{N} \|_{1\rightarrow \infty}. \]
It was also observed by Hastings~\cite{hastings09} that picking the subspace corresponding to $\mathcal{N}$ by replacing the random unitary $U$ with an operator picked according to Haar measure on the orthogonal group $O(d_A d_B)$  gives a channel such that, when $r \approx d_A = d_B$,
\[ \| \mathcal{N}^{\otimes 2} \|_{1\rightarrow \infty} \approx \| \mathcal{N} \|_{1\rightarrow \infty}. \]
Thus Theorem \ref{thm:intro} is perhaps essentially the strongest multiplicativity result one could expect for the maximum output $\infty$-norm of random quantum channels (although note that it is not in fact known whether multiplicativity is violated by two copies of a random channel $\mathcal{N}$).

We remark that our results also apply without change to the setting where, instead of taking $n$ copies of a fixed random channel $\mathcal{N}$, one takes $n$ independently picked random channels $\mathcal{N}_1,\dots,\mathcal{N}_n$ (where $n$ is fixed) and compares $\|\mathcal{N}_1 \otimes \dots \otimes \mathcal{N}_n\|_{1 \rightarrow p}$ with $(\|\mathcal{N}_1\|_{1\rightarrow p} \dots \|\mathcal{N}_n\|_{1\rightarrow p})^{\alpha}$.


\subsection{The case of the von Neumann entropy}

In certain regimes, our results also imply a weak {\em additivity} result for the von Neumann entropy\footnote{I would like to thank Fernando Brand\~{a}o for pointing this out.}. Recall that for a quantum channel $\mathcal{N}$, the minimum output R\'enyi $\infty$-entropy of $\mathcal{N}$ is defined by taking the limit as $H_{\infty}^{\min}(\mathcal{N}) = -\log \|\mathcal{N}\|_{1\rightarrow \infty}$. A corollary of Theorem \ref{thm:main} below is that there is a universal constant $C$ such that the probability that the regularised minimum output $\infty$-entropy of a random quantum channel $\mathcal{N}$ does {\em not} satisfy
\[
\frac{1}{n} H_{\infty}^{\min}(\mathcal{N}^{\otimes n}) \ge 
\begin{cases}
\frac{1}{2}\left( \log d_A + \log d_B - \log r \right) - C & \text{if } r \ge d_B/d_A\\
\log d_A - C & \text{if } r \le d_B/d_A
\end{cases}
\]
is exponentially small in $m:=\min\{r,d_A,d_B\}$, assuming that $m \ge 2 (\log_2 \max\{d_A,d_B\})^{3/2}$. As $\frac{1}{n} H^{\min}(\mathcal{N}^{\otimes n}) \ge \frac{1}{n} H^{\min}_{\infty}(\mathcal{N}^{\otimes n})$ for all channels $\mathcal{N}$ by monotonicity of R\'enyi entropies, any upper bound on $H^{\min}(\mathcal{N})$ implies a limitation on the extent of possible additivity violations for $H^{\min}(\mathcal{N})$. In particular, as $H^{\min}(\mathcal{N})$ is always upper bounded by $\log d_A$, we obtain a weak additivity result for the minimum output von Neumann entropy when $r$ is not too large with respect to $d_B$. For example, if $r = d_A = d_B$, we obtain that
\[ \frac{1}{n} H^{\min}(\mathcal{N}^{\otimes n}) \ge \frac{1}{2} H^{\min}(\mathcal{N}) - C \]
with high probability. If $r \le d_B/d_A$, we get the stronger result that with high probability
\[ \frac{1}{n} H^{\min}(\mathcal{N}^{\otimes n}) \ge H^{\min}(\mathcal{N}) - C. \]


\subsection{Organisation and proof strategy}

The remainder of this paper is devoted to the proof of Theorem \ref{thm:intro}. Conceptually, the proof is simple: we find a general upper bound on $h_{\operatorname{SEP}}(M)$ (for arbitrary $M$) which is multiplicative, and then show that this upper bound is not too far from the truth when $M$ corresponds to a random quantum channel. The upper bound we use is the operator norm of the partial transpose of $M$, $\|M^\Gamma\|_\infty$, and the main technical contribution of this paper is to prove tail bounds on this quantity, which is carried out using the method of moments from random matrix theory. Note that we have not attempted to optimise the constants and lower order terms which occur in our bounds, which could probably be substantially improved.

The next section discusses the proof strategy and some related work, and states a more formal version of Theorem \ref{thm:intro}. The following section contains the proofs of the main technical results and lemmas. The paper finishes in Section \ref{sec:conclusion} with some conclusions and suggestions for future work.


\section{Maximum overlap with separable states and PPT}
\label{sec:proofstart}

Our proof of Theorem \ref{thm:intro} will be based on a general upper bound strategy for $h_{\text{SEP}}(M)$. Maximising over the set of separable states is a daunting task, and a useful relaxation is to maximise over the larger set of PPT states (bipartite quantum states $\rho$ such that $\rho^{\Gamma} \ge 0$, where $^\Gamma$ denotes the partial transpose operation, i.e.\ the transpose operation performed only on the second subsystem) and consider 
\[ h_{\text{PPT}}(M) := \max_{\rho \in \text{PPT}} \tr M \rho. \]
An upper bound can be put on this quantity as follows.
\begin{prop}
$h_{\operatorname{PPT}}(M) \le \|M^\Gamma\|_\infty$.
\end{prop}

\begin{proof}
By definition, we have
\[ h_{\text{PPT}}(M) = \max_{\substack{\rho,\;\rho\ge 0,\\\rho^\Gamma\ge 0,\;\tr \rho=1}} \tr M \rho = \max_{\substack{\sigma,\;\sigma\ge 0,\\\sigma^\Gamma\ge 0,\;\tr \sigma^\Gamma=1}} \tr M \sigma^\Gamma, \]
and for any density matrix $\sigma$, $\tr M \sigma^\Gamma = \tr M^\Gamma \sigma \le \|M^\Gamma\|_\infty$.
\end{proof}

This upper bound is often very weak (and for example may be much bigger than 1), but for some operators $M$ it does give something non-trivial. Observe that if $M$ projects onto a subspace $S \subseteq \C^{d_A} \otimes \C^{d_B}$, $\|M^\Gamma\|_\infty$ is {\em not} the same as $\max_{\ket{\psi} \in S} \|\psi^\Gamma\|_\infty$. Indeed, for any state $\ket{\psi}$ we have
\[ \|\psi^\Gamma \|_\infty \le \|\psi^\Gamma\|_2 = \|\psi\|_2 = 1, \]
whereas $\|M^\Gamma\|_\infty$ may be considerably higher.

A key property of $\|M^\Gamma\|_\infty$ which we will use is that it is multiplicative.

\begin{obs}
For any operators $M$, $N$,
$\|(M \otimes N)^\Gamma\|_\infty = \|M^\Gamma \otimes N^\Gamma\|_\infty = \|M^\Gamma\|_\infty \|N^\Gamma\|_\infty$.
\end{obs}

Thus, if we can show that $\|M^\Gamma\|_\infty \le \delta$ for some $\delta$, we immediately have that $h_{\text{SEP}}(M^{\otimes n}) \le \delta^n$. If $\delta$ is small enough, this can be used to prove that $M$ obeys weak $\infty$-norm multiplicativity. We formalise this as the following observation.

\begin{obs}
For any projector $M$, if $\|M^\Gamma\|_\infty \le h_{\operatorname{SEP}}(M)^{\alpha}$ for some constant $0 < \alpha \le 1$, $M$'s corresponding quantum channel obeys weak $\infty$-norm multiplicativity with exponent $\alpha$.
\end{obs}

Of course, for many operators $M$, $\|M^\Gamma\|_\infty$ is much larger than 1. However, we will see that for random quantum channels, with high probability this strategy does lead to a non-trivial bound. Let $M$ be the projector onto a subspace picked at random from the set of dimension $r$ subspaces of $\C^{d_A} \otimes \C^{d_B}$, i.e.\ $M$ is formed by taking the projector $M_0$ onto a fixed dimension $r$ subspace $S_0 \subseteq \C^{d_A} \otimes \C^{d_B}$, and conjugating $M_0$ by a Haar-random unitary. Then we have the following easy {\em lower} bound on $h_{\operatorname{SEP}}(M)$.

\begin{prop}
\label{prop:genlower}
Let $M$ be the projector onto an $r$-dimensional subspace of $\C^{d_A} \otimes \C^{d_B}$. Then
\[ h_{\operatorname{SEP}}(M) \ge \max \left\{ \frac{r}{d_A d_B}, \frac{1}{d_A} \right\}. \]
\end{prop}

\begin{proof}
For the first part, pick a random product state $\ket{\psi} = \ket{\psi_A} \otimes \ket{\psi_B}$ by choosing each of  $\ket{\psi_A}$ and $\ket{\psi_B}$ uniformly at random (according to Haar measure). On average this achieves
\[ \E_{\psi}[\tr M \psi] = \tr M (\E_{\psi_A} \psi_A) \otimes (\E_{\psi_B} \psi_B) = \frac{r}{d_A d_B}. \]
The second part follows from the correspondence with quantum channels. Any state output from the channel which corresponds to $M$ must have largest eigenvalue at least $1/d_A$, so $h_{\operatorname{SEP}}(M) \ge 1/d_A$.
\end{proof}

On the other hand, we will prove the following upper bound.

\begin{thm}
\label{thm:main}
Fix $d_A \le d_B$ and $r$ such that $m := \min\{d_A,r\} \ge 2 (\log_2 d_B)^{3/2}$, and let $M$ be the projector onto a random dimension $r$ subspace of $\C^{d_A} \otimes \C^{d_B}$. Then there exists a universal constant $C$ such that, if $r \ge d_B / d_A$,
\[ \E \|M^{\Gamma}\|_\infty \le \frac{Cr^{1/2}}{ d_A^{1/2} d_B^{1/2}}, \]
and if $r \le d_B / d_A$,
\[ \E \|M^{\Gamma}\|_\infty \le \frac{C}{d_A}. \]
Further, for any $\delta > 0$, there exists a universal constant $C'$ such that if $r \ge d_B / d_A$ then
\[ \Pr\left[\|M^{\Gamma}\|_\infty \ge \delta \frac{2^8 r^{1/2}}{d_A^{1/2} d_B^{1/2}} \right] \le C' m^{16/3} \delta^{-(m/2)^{2/3}}, \]
and if $r \le d_B / d_A$ then
\[ \Pr\left[\|M^{\Gamma}\|_\infty \ge \delta \frac{2^8}{d_A} \right] \le C' m^{16/3} \delta^{-(m/2)^{2/3}}. \]
\end{thm}

Combining Proposition~\ref{prop:genlower} and Theorem \ref{thm:main}, we obtain Theorem \ref{thm:intro}, which we now state more formally. 

\begin{thm}
\label{thm:toplevel}
Fix $d_A \le d_B$ and $r$ such that $m := \min\{d_A,r\} \ge 2 (\log d_B)^{3/2}$, and let $M$ be the projector onto a random dimension $r$ subspace of $\C^{d_A} \otimes \C^{d_B}$. Then there exists a universal constant $C$ such that if $r \ge d_B/d_A$, then
\[ \Pr\left[ \|M^{\Gamma}\|_\infty \ge h_{\operatorname{SEP}}(M)^{1/2 - \epsilon} \right] \le C m^{16/3} 2^{-(m/2)^{2/3}}, \]
where $\epsilon = \frac{9}{\log_2(d_A d_B/r)}$, and if $r \le d_B/d_A$, then
\[ \Pr\left[ \|M^{\Gamma}\|_\infty \ge h_{\operatorname{SEP}}(M)^{1 - \epsilon'} \right] \le C m^{16/3} 2^{-(m/2)^{2/3}}, \]
where $\epsilon' = \frac{9}{\log_2 d_A}$.
\end{thm}

\begin{proof}
If $r \ge d_B / d_A$, using Proposition \ref{prop:genlower} and taking $\delta=2$ in Theorem \ref{thm:main}, we obtain
\beas
\Pr[ \|M^{\Gamma}\|_\infty \ge h_{\operatorname{SEP}}(M)^{1/2 - \epsilon}] &\le& \Pr\left[ \|M^{\Gamma}\|_\infty \ge \left(\frac{r}{d_A d_B}\right)^{1/2 - \epsilon}\right]\\
&=& \Pr\left[ \|M^{\Gamma}\|_\infty \ge \frac{2^9 r^{1/2}}{d_A^{1/2} d_B^{1/2}}\right]\\
&\le& C' m^{16/3} 2^{-(m/2)^{2/3}},
\eeas
and a similar argument holds for $r \le d_B / d_A$.
\end{proof}


\subsection{Related work}

Important intuition that some version of Theorem \ref{thm:main} should hold comes from recent work by Aubrun~\cite{aubrun10b}, who studied the asymptotic spectrum of partially transposed Wishart matrices (also see~\cite{collins11} for some interesting related work). Such matrices, which are a natural model for random mixed quantum states, can be formed as follows. Let $G$ be a $d \times r$ matrix whose entries are picked independently from the complex normal distribution $N(0,1)$, and set $W = \frac{1}{d} G G^\dag$; we say that $W$ is a $(d,r)$-Wishart matrix\footnote{Note that we use a different normalisation to~\cite{aubrun10b}.}. Then one of Aubrun's results can be stated as follows. For a fixed constant $0<\alpha<1$, let $W_d$ be a $(d^2,\lfloor \alpha d^2\rfloor)$-Wishart matrix, understood as acting on the bipartite space $\C^d \otimes \C^d$, and set $Y_d = W_d^\Gamma$. Then
\[ \lim_{d\rightarrow \infty} \Pr[|\lambda_{\max}(Y_d) - \sqrt{\alpha}(2 + \sqrt{\alpha})| > \epsilon] = 0. \]
As $r = \alpha d^2$, this implies that with high probability, $\lambda_{\max}(Y_d) = O(\sqrt{r}/d)$ for large $d$. As the columns of $G$ are approximately orthogonal for large $d$, one might expect that the operator norm of the partial transpose of the projector onto a random $r$-dimensional subspace of $\C^d \otimes \C^d$ should behave similarly to that of $Y_d$; the present work demonstrates that this is indeed true.

As well as the final bound we obtain being of a similar order to that of Aubrun, the reader familiar with~\cite{aubrun10b} will notice that we need to prove some analogous combinatorial lemmas. However, it does not seem clear that the results given here could be obtained as a formal consequence of~\cite{aubrun10b}, or indeed vice versa; it would be interesting to determine whether this is the case. We also note two minor technical differences between this work and~\cite{aubrun10b}: Aubrun's result is only stated for $d_A = d_B = d$ (some recent work by Banica and Nechita removes this restriction~\cite{banica11}), and only for $\alpha$ constant (so $r$ grows as a constant fraction of $d$). By contrast, here $d_A$, $d_B$ and $r$ can be essentially arbitrary, although the bound we obtain becomes trivial if $r$ is too large as a fraction of $d_A d_B$.

Some very recent work by Collins, Fukuda and Nechita~\cite{collins11b} also uses related techniques to those which we use here to prove Theorem \ref{thm:main} (e.g.\ calculations with Weingarten functions~\cite{collins06}, see Section \ref{sec:weingarten} below). The goal of~\cite{collins11b} was to find the state which, when input to a tensor product of two random quantum channels (either the same channel, or a channel and its conjugate), achieves minimal output entropy. A sequence of recent papers by Collins and Nechita~\cite{collins10b,collins10,collins11d,collins11c} carries out a number of interesting quantum information-theoretic calculations using Weingarten functions. This method seems to be a powerful tool which may be expected to find many other applications in quantum information.

The rest of this paper is devoted to the proof of Theorem \ref{thm:main}, which we now begin.


\section{Proof of Theorem \ref{thm:main}: moments of partially transposed projectors}

As before, let $M_0$ be the projector onto an arbitrary fixed $r$-dimensional subspace of $\C^{d_A} \otimes \C^{d_B}$, where $d_A \le d_B$, and let $M = U M_0 U^\dag$ be the result of applying a Haar-random unitary operator to $M_0$. Also let
\[ D_d(\pi) := \sum_{i_1,\dots,i_k=1}^d \ket{i_{\pi(1)}}\ket{i_{\pi(2)}}\dots\ket{i_{\pi(k)}} \bra{i_1}\dots\bra{i_k} \]
be the representation of the permutation $\pi \in S_k$ which acts on $k$ systems of dimension $d$ by permuting the systems; in the case $d = d_A d_B$, we simply omit $d$ and write $D(\pi) := D_{d_A d_B}(\pi)$. Let $\kappa\in S_k$ be the permutation that maps $i \mapsto i+1$ for $i=1,\dots,k-1$, and maps $k \mapsto 1$. For any permutation $\pi$, let $c(\pi)$ be the number of cycles in $\pi$.

In order to put good upper bounds on $\|M^{\Gamma}\|_\infty$, it suffices to understand $\E \tr[(M^{\Gamma})^k]$ for arbitrary even $k$. Observe that
\beas \E \tr[(M^{\Gamma})^k] &=& \E_U \tr[((U M_0 U^{\dag})^{\Gamma})^k]\\
&=& \tr[\E_U D(\kappa)((U^{\otimes k} M_0^{\otimes k} (U^{\dag})^{\otimes k})^{\Gamma})]\\
&=& \tr[D(\kappa)^{\Gamma} M^{(k)}],
\eeas
where for brevity we write
\[ M^{(k)} := \E_U[U^{\otimes k} M_0^{\otimes k} (U^{\dag})^{\otimes k}]. \]
The second equality above is the observation that $\tr[X^k] = \tr D(\kappa) X^{\otimes k}$ (indeed, this holds for any permutation $\pi \in S_k$ such that $c(\pi)=1$).

Theorem \ref{thm:main} will follow easily from the following result.

\begin{thm}
\label{thm:kbound}
There is a universal constant $C$ such that, for any $k$ satisfying $2 k^{3/2} \le \min\{d_A,r\}$,
\[ \tr[D(\kappa)^{\Gamma} M^{(k)}] \le \begin{cases}
C k^8 2^{6k} r^{k/2} d_A^{-k/2+1} d_B^{-k/2+1} & \text{ if } r \ge d_B/d_A\\
C k^8 2^{6k}  d_A^{-k+1} d_B & \text{ otherwise.}
\end{cases}
\]
\end{thm}

\begin{proof}[Proof of Theorem \ref{thm:main} (assuming Theorem \ref{thm:kbound})]
For any even $k \ge 2$,
\[ \E \|M^{\Gamma}\|_\infty \le (\E[\tr (M^{\Gamma})^k])^{1/k}, \]
and taking $k$ to be the largest even number smaller than $(m/2)^{2/3}$, where $m := \min\{d_A,r\}$ and we assume that $m \ge 2 (\log_2 d_B)^{3/2}$, we obtain that there is a universal constant $C'$ such that
\[ \E \|M^{\Gamma}\|_\infty \le \frac{C' r^{1/2}}{d_A^{1/2} d_B^{1/2}} \text{ if } r \ge d_B/d_A, \]
and
\[
\E \|M^{\Gamma}\|_\infty \le \frac{C'}{d_A} \text{  otherwise,}
\]
which is the first part of Theorem \ref{thm:main}. (Note that we took $m \ge 2 (\log_2 d_B)^{3/2}$ in order to kill off terms of the form $d_B^{1/k}$.) The concentration bound follows from Markov's inequality, which implies that, for any even $k \ge 2$ and any $x > 0$,
\[ \Pr[\|M^{\Gamma}\|_\infty \ge x] = \Pr\left[\|M^{\Gamma}\|_\infty^k \ge x^k\right] \le \frac{\E[\tr (M^{\Gamma})^k]}{x^k}; \]
and once again taking $k$ to be the largest even number smaller than $(m/2)^{2/3}$, we obtain
\[ \Pr\left[\|M^{\Gamma}\|_\infty \ge \delta \frac{2^8 r^{1/2}}{d_A^{1/2} d_B^{1/2}} \right] \le \frac{C k^8 d_A d_B }{(4\delta)^k} \le C' m^{16/3} \delta^{-(m/2)^{2/3}} \]
if $r \ge d_B/d_A$, and similarly
\[ \Pr\left[\|M^{\Gamma}\|_\infty \ge \delta \frac{2^8}{d_A} \right] \le \frac{C k^8 d_A d_B}{(4\delta)^k} \le C' m^{16/3} \delta^{-(m/2)^{2/3}} \]
if $r \le d_B/d_A$.
\end{proof}

We now proceed to prove Theorem \ref{thm:kbound}. By Schur-Weyl duality, as $M^{(k)}$ commutes with local unitaries and permutations of the $k$ systems, it can be expanded in terms of permutations $D(\pi)$ as
\[ M^{(k)} = \sum_{\pi \in S_k} \alpha_{\pi} D(\pi) \]
for some coefficients $\alpha_\pi$. When $k$ is small with respect to $d_A d_B$, the matrices $\{D(\pi)\}$ are almost orthonormal with respect to the normalised Hilbert-Schmidt inner product, i.e.\
\[ \frac{1}{(d_A d_B)^k} \tr[D(\pi)^\dag D(\sigma)] \approx 0 \text{ if } \pi \neq \sigma. \]
It has been pointed out by Harrow that one can prove a number of interesting  results in quantum information based on this philosophy~\cite{harrow12}. In our case, because of this near-orthonormality we ought to have
\[ \alpha_{\pi} \approx \frac{\tr [M^{(k)} D(\pi^{-1})]}{\tr [D(\pi^{-1})D(\pi)]} = \frac{r^{c(\pi)}}{(d_A d_B)^k}; \]
the following key technical lemma, which we prove in Section \ref{sec:weingarten}, makes this approximate equality quantitative.

\begin{lem}
\label{lem:alphabound}
Assume $k \le (r/2)^{2/3}$. Then there is a universal constant $C$ such that
\[ |\alpha_\pi| \le C k 2^{4k} \frac{r^{c(\pi)}}{(d_A d_B)^k}. \]
\end{lem}

Now it holds that
\bea
\nonumber \tr[D(\kappa)^{\Gamma} M^{(k)}] &=& \sum_{\pi \in S_k} \alpha_{\pi} \tr[D(\kappa)^{\Gamma} D(\pi) ] = \sum_{\pi \in S_k} d_A^{c(\kappa \pi)} d_B^{c(\kappa^{-1} \pi)} \alpha_{\pi}\\
\label{eq:sumbound} &\le& C k 2^{4k} \sum_{\pi \in S_k} d_A^{c(\kappa \pi) - k} d_B^{c(\kappa^{-1} \pi) - k} r^{c(\pi)},
\eea
where the second equality is the fact that
\beas
\tr[D(\kappa)^{\Gamma} D(\pi)] &=& \tr[(D_{d_A}(\kappa) \otimes D_{d_B}(\kappa)^T) (D_{d_A}(\pi) \otimes D_{d_B}(\pi))]\\
&=& \tr[D_{d_A}(\kappa)D_{d_A}(\pi)] \tr[D_{d_B}(\kappa^{-1})D_{d_B}(\pi)]\\
&=& d_A^{c(\kappa \pi)} d_B^{c(\kappa^{-1}\pi)},
\eeas
and the inequality is Lemma \ref{lem:alphabound}. In order to upper bound the quantity (\ref{eq:sumbound}), we will use the following lemma\footnote{Related combinatorial results appear in the literature, for example in~\cite{aubrun10b}. However, the precise statement we need here does not seem to have been written down.}, which we also prove afterwards, in Section \ref{sec:comb}.

\begin{lem}
\label{lem:cycgenub}
For any $\pi,\sigma \in S_k$,
\[ c(\pi^{-1}\sigma) + c(\sigma) \le k + c(\pi). \]
Further, for any $\pi \in S_k$ and any integer $\delta \ge 0$, there are at most $4^{k-1} k^{3\delta/2 + 1}$ permutations $\sigma \in S_k$ such that
\[ c(\pi^{-1}\sigma) + c(\sigma) = k + c(\pi) - \delta. \]
\end{lem}

The exact constants in this lemma are not so important; what matters is that, when $i+j \approx k + c(\pi)$, there are not too many permutations $\sigma$ such that $c(\pi^{-1}\sigma) = i$, $c(\sigma) = j$ (i.e.\ there are $O(4^k \poly(k))$ of them, rather than $\Omega(k!)$). Now define
\[ N(a,b,c) := |\{\pi \in S_k:c(\kappa \pi)=a, c(\kappa^{-1}\pi) = b, c(\pi) = c\}| \]
and rewrite (\ref{eq:sumbound}) as
\[ \tr[D(\kappa)^{\Gamma} M^{(k)}] \le \frac{C k 2^{4k} }{d_A^k d_B^k} \sum_{a,b,c \in \{1,\dots,k\}} N(a,b,c) d_A^{a} d_B^{b} r^{c}. \]
By Lemma \ref{lem:cycgenub}, we have $N(a,b,c)=0$ whenever
\[ a + b > k + 2,\; a + c > k+1,\;\text{or } b + c > k+1. \]
Call the triple $(a,b,c)$ {\em valid} when none of these events occur. Further, we have the bound from Lemma \ref{lem:cycgenub} that for all valid triples,
\[ N(a,b,c) \le 4^{k-1} k^{(3/2)(k + 2 - \max\{a+b,a+c,b+c\} )+1}. \]
We therefore obtain the upper bound
\[ \tr[D(\kappa)^{\Gamma} M^{(k)}] \le \frac{C k^4 2^{4k} }{d_A^k d_B^k} \max_{(a,b,c) \text{ valid}} \left\{ 4^k k^{(3/2)(k - \max\{a+b,a+c,b+c\} )+4} d_A^{a} d_B^{b} r^{c} \right\}. \]
Now observe that, as long as $k^{3/2} \le \min\{d_A, r\}$, this maximum will be achieved when at least one of the validity inequalities is saturated, because otherwise we could increase the maximum by at least a factor of $\min\{d_A,r\}/k^{3/2}$, by increasing at least one of $(a,b,c)$. Thus the upper bound simplifies to
\be \label{eq:validbound} \tr[D(\kappa)^{\Gamma} M^{(k)}] \le \frac{C k^8 2^{6k} }{d_A^k d_B^k} \max_{(a,b,c) \text{ valid}} \left\{ d_A^{a} d_B^{b} r^{c} \right\}. \ee
We have seen that the triple $(a,b,c)$ satisfies the linear inequalities $2 \le a + b \le k+2$, $2 \le a + c \le k+1$, $2 \le b + c \le k+1$. Therefore, the optimal value of the following simple linear program gives an upper bound on $\log \max_{(a,b,c) \text{ valid}} d_A^{a} d_B^{b} r^{c}$:
\beas
\text{maximise} && (\log d_A) a + (\log d_B) b + (\log r) c\\
\text{subject to} && 2 \le a + b \le k+2\\
&& 2 \le a + c \le k+1\\
&& 2 \le b + c \le k+1,\\
&& a,b,c \ge 0.
\eeas
Performing standard manipulations gives the dual:
\beas
\text{minimise} && (k+1)u^+ + (k+1) v^+ + (k+2) w^+ -2 (u^- + v^- + w^-)\\
\text{subject to} && u + w \ge \log d_A \\
 && v + w \ge \log d_B \\
  && u + v \ge \log r,
\eeas
where the notation $w^+$ (resp.\ $w^-$) is used for the positive (resp.\ negative) part of $x$, i.e.\ $w^+ = \max \{ w,0 \}$, $w^- = -\min \{ w,0 \}$. Observe that the following dual solution saturates all three inequalities:
\[ u = \frac{1}{2}\left(\log d_A - \log d_B + \log r \right),\; v = \frac{1}{2}\left(-\log d_A + \log d_B + \log r \right),\; w = \frac{1}{2}\left(\log d_A + \log d_B - \log r \right). \]
As $r \le d_A d_B$ and $d_A \le d_B$,  in this solution $v$ and $w$ are always non-negative. If it additionally holds that $r \ge d_B / d_A$, so $\log d_B \ge \log d_A + \log r$, $u$ is also non-negative, so we achieve an objective value of
\[ (k+1) u + (k+1) v + (k+2)w = \frac{1}{2}\left( (k+2) \log d_A + (k+2) \log d_B + k \log r \right). \]
On the other hand, if $r \le d_B/d_A$, consider the solution
\[ u = 0,\; v = \log d_B - \log d_A,\; w = \log d_A. \]
It is easy to verify that $u,v,w \ge 0$ and this solution achieves an objective value of
\[ \log d_A + (k+1)\log d_B. \]
These two solutions correspond to upper bounds in (\ref{eq:validbound}) of
\[ \tr[D(\kappa)^{\Gamma} M^{(k)}] \le \begin{cases}
C k^8 2^{6k} r^{k/2} d_A^{-k/2+1} d_B^{-k/2+1} & \text{ if } r \ge d_B/d_A\\
C k^8 2^{6k}  d_A^{-k+1} d_B & \text{ otherwise.}
\end{cases}
\]
This completes the proof of Theorem \ref{thm:kbound}.


\subsection{Proof of Lemma \ref{lem:cycgenub}: combinatorics of permutations}
\label{sec:comb}

Our next task is to prove the above combinatorial lemma, which we restate for convenience.

\begin{replem}{lem:cycgenub}
For any $\pi,\sigma \in S_k$,
\[ c(\pi^{-1}\sigma) + c(\sigma) \le k + c(\pi). \]
Further, for any $\pi \in S_k$ and any integer $\delta \ge 0$, there are at most $4^{k-1} k^{3\delta/2 + 1}$ permutations $\sigma \in S_k$ such that
\[ c(\pi^{-1}\sigma) + c(\sigma) = k + c(\pi) - \delta. \]
\end{replem}

The proof will rely on some previously known combinatorial results regarding permutations; we first review some basic ideas in this area. Consider the Cayley graph whose vertices are elements of $S_k$ and where two vertices $\pi$, $\sigma$ are connected if and only if there exists a transposition $\tau$ such that $\tau \pi = \sigma$. For any permutations $\pi$, $\sigma$, let $d(\pi,\sigma)$ be the shortest-path metric with respect to this graph (i.e.\ the minimum number of transpositions required to change $\pi$ into $\sigma$). This is indeed a metric as $d(\pi,\sigma) \le d(\pi,\rho) + d(\rho,\sigma)$ for any $\rho \in S_k$. Also observe that $d(\pi,\sigma) = d(\sigma,\pi)$ and $d(\pi,\sigma) = d(\pi^{-1} \sigma,e)$, where $e$ denotes the identity permutation.

Let $\tau$ be a transposition exchanging elements $i$ and $j$. If $\pi$ is given in cycle notation as $\pi = (c_1)(c_2)\dots(c_ \ell)$, where each $c_p$ is a sequence of integers, then the permutation $\tau \pi$ takes one of two forms depending on whether the transposed elements are in the same cycle or not. Let $c=(c_1,\dots,c_m)$ and $d=(d_1,\dots,d_{m'})$ be two cycles in $\pi$. If elements $c_i$ and $c_j$ are transposed, $c$ splits into two cycles $(c_1,\dots,c_{i-1},c_j,\dots,c_m)$ and $(c_i,\dots,c_{j-1})$. If elements $c_i$ and $d_j$ are transposed, cycles $c$ and $d$ are joined to produce the cycle $(c_1,\dots,c_{i-1},d_j,d_{j+1},\dots,d_{m'},d_1,\dots,d_{j-1},c_i,\dots,c_m)$. Thus performing a transposition $\tau$ always either increases or decreases the number of cycles by exactly 1, and in each case the result preserves the ordering of elements within the original cycles.

This implies that, for any $\pi \in S_k$, $d(\pi,e) = k - c(\pi)$, because $\ell-1$ transpositions are required to split a cycle of length $\ell$ into $\ell$ cycles of length 1, and this can be achieved.  This immediately implies the first part of Lemma \ref{lem:cycgenub} (which is well-known). Simply use the triangle inequality:
\be \label{eq:triangleineq} c(\pi^{-1}\sigma) + c(\sigma) = k - d(\pi,\sigma) + k - d(e,\sigma) \le 2k - d(e,\pi) = k + c(\pi). \ee
Also note that, if $(\tau_1,\dots,\tau_m)$ is a sequence of transpositions such that
\[ \tau_m \dots \tau_1 \pi = \sigma, \]
it holds that $m - d(\pi,\sigma)$ is a multiple of 2.

We now turn to the second part of Lemma \ref{lem:cycgenub}. We first observe that the special case $\delta=0$ has a simple (and well-known, e.g.\ see~\cite{nica06}) proof. 

\begin{lem}
\label{lem:critical}
For any permutation $\pi \in S_k$,
\[ |\{\sigma: c(\pi^{-1} \sigma) + c(\sigma) = k + c(\pi)\}| \le C_k, \]
where $C_k$ is the $k$'th Catalan number
\[ C_k := \frac{1}{k+1} \binom{2k}{k}. \]
\end{lem}

\begin{proof}
The number of permutations $\sigma \in S_k$ such that $c(\pi^{-1} \sigma) + c(\sigma) = k+c(\pi)$ is equal to
\be \label{eq:critbound} |\{\sigma: d(\pi,\sigma) + d(\sigma,e) = d(\pi,e) \}|, \ee
or in other words the same as the number of permutations which lie on a shortest path (``geodesic'') between $\pi$ and $e$. 
Any such permutation must be obtained from $\pi$ by a sequence of transpositions, each of which splits a cycle of $\pi$ in two.
Permutations on the shortest path between a cycle of length $\ell$ and the identity are known to be isomorphic to non-crossing partitions of $\{1,\dots,\ell\}$~\cite{biane97}, which are counted by the Catalan numbers $C_\ell$. 
Therefore, an overall upper bound on (\ref{eq:critbound}) is
\[ \prod_{i=1}^{c(\pi)} C_{c_i} \le C_k, \]
where $c_i$ is the length of the $i$'th cycle of $\pi$ and we use the simple upper bound $C_a C_b \le C_{a+b}$, valid for all integers $a,b \ge 1$.
\end{proof}

For the more complicated case of $\delta>0$, we will rely on a result of Adrianov~\cite{adrianov97}. This work enumerates the so-called bicoloured unicellular maps with $m$ white and $n$ black vertices and $k$ edges, one of which is marked; this is known to be equivalent to the problem of counting pairs of permutations $\pi,\sigma \in S_k$ such that $c(\pi)=m$, $c(\sigma)=n$ and $\pi \sigma$ is a fixed cycle of length $k$~\cite{adrianov97}. For $\pi \in S_k$, write
\[ N_g(\pi) := |\{\sigma: c(\pi \sigma^{-1}) + c(\sigma) = k + c(\pi) - 2g\}| = |\{\sigma: d(\pi,\sigma)+ d(\sigma,e) = d(\pi,e) + 2g\}|. \]
By Lemma \ref{lem:critical} and the above arguments, $N_g(\pi) = 0$ for all $g<0$ and for non-integer $g$. Also set $B_g(k) := N_g(\kappa)$, where $\kappa \in S_k$ is an arbitrary cycle of length $k$. Then the following result holds.

\begin{thm}[Adrianov~\cite{adrianov97}, Corollary 3]
\label{thm:recur}
$B_g(k)$ satisfies the recurrence
\[ (k+1) B_g(k) = 2(2k-1)B_g(k-1) + (k-2)(k-1)^2 B_{g-1}(k-2) \]
with initial conditions
\[ B_0(1) = 1, \; B_0(2) = 2,\; B_g(k) = 0 \text{ for } g<0. \]
\end{thm}

Using Theorem \ref{thm:recur}, we now complete the proof of Lemma \ref{lem:cycgenub}. We have
\[ B_g(k) \le 4 B_g(k-1) + (k-1)^2 B_{g-1}(k-2), \]
and we now apply induction on $g$ to upper bound this quantity. For $g=0$, we have $B_0(k) \le 4^{k-1}$ (agreeing with the known exact values $B_0(k) = C_k$). For any $g>0$ we have $B_g(1) = 0$, and more generally
\beas
B_g(k) &\le& 4 \left(4 B_g(k-2) + (k-2)^2 B_{g-1}(k-3) \right) + (k-1)^2 B_{g-1}(k-2)\\
&\le& \dots \le \sum_{i=1}^{k-1} 4^{k-i-1} i^2 B_{g-1}(i-1).
\eeas
We now use the inductive hypothesis that $B_{g-1}(k) \le 4^{k-1} k^{3(g-1)}$ for all $k$, which implies that
\beas
B_g(k) &\le& \sum_{i=1}^{k-1} 4^{k-i-1} i^2 4^{i-2} (i-1)^{3(g-1)} \le 4^{k-3} \sum_{i=1}^{k-1} i^{3g-1}\\
&\le& 4^{k-3} \int_{1}^{k} x^{3g-1} dx \le 4^{k-3} \frac{k^{3g}}{3g} \le 4^{k-1} k^{3g}.
\eeas
We use this to show that, for any $k\ge 2$, any $\pi \in S_k$ and any permutation $\kappa \in S_k$ which consists of one cycle of length $k$,
\be \label{eq:sumineq} \sum_{h = 0}^g N_h(\pi) \le \sum_{h = 0}^g N_h(\kappa) \le 4^{k-1} k^{3g+1}. \ee
The argument for the first inequality is as follows. Let $\kappa'$ be a cycle of length $k$ which minimises $d(\kappa',\pi)$. Then, for any permutation $\sigma$ such that $d(\pi,\sigma)+ d(\sigma,e) = d(\pi,e) + 2g$, it holds that
\[ d(\kappa',\sigma) + d(\sigma,e) \le d(\kappa',\pi) + d(\pi,\sigma) + d(\sigma,e) = d(\kappa',\pi) + d(\pi,e) + 2g = d(\kappa',e) + 2g, \]
where the second equality holds because $\pi$ is on the shortest path between $\kappa'$ and $e$. Thus $\sigma$ contributes to $N_h(\kappa')$ for some $h \le g$. Just looking at one term in the sum on the left-hand side of (\ref{eq:sumineq}), we have that $N_g(\pi) \le  4^{k-1} k^{3g+1}$ for all $\pi \in S_k$. To complete the proof of Lemma \ref{lem:cycgenub}, simply observe that
\[ |\{\sigma: c(\pi^{-1}\sigma) + c(\sigma) = k + c(\pi) - \delta\}| = N_{\delta/2}(\pi) \]
for even $\delta$, and $|\{\sigma: c(\pi^{-1}\sigma) + c(\sigma) = k + c(\pi) -\delta\}| = 0$ for odd $\delta$.


\subsection{Proof of Lemma \ref{lem:alphabound}: permutations and Weingarten functions}
\label{sec:weingarten}

Let $A$ be the symmetric matrix defined by $A_{\pi \sigma} = d^{c(\pi^{-1}\sigma)-k}$, for $\pi,\sigma \in S_k$. Given some matrix $M$ such that $M = \sum_{\pi \in S_k} \alpha_{\pi} D_d(\pi)$, $A$ determines the coefficients $\alpha_{\pi}$ as follows:
\[ \tr M D_d(\sigma) = \sum_{\pi \in S_k} \alpha_{\pi} d^{c(\pi \sigma)} = d^k \sum_{\pi \in S_k} A_{\sigma^{-1}\pi} \alpha_{\pi}. \]
Letting $v$ and $w$ be the vectors defined by
\[ v_{\sigma} = \frac{1}{d^k} \tr M D_d(\sigma^{-1}),\;w_{\pi} = \alpha_{\pi}, \]
this is equivalent to the claim that $Aw = v$. Thus, if $A^{-1}$ exists, we can determine the $\alpha_{\pi}$ coefficients by computing $A^{-1}v$. Note that $A$ is approximately equal to the identity when $d$ is large with respect to $k$, as its off-diagonal entries rapidly decay~\cite{harrow12}.

In order to evaluate the entries of $A^{-1}$, we define the {\em Weingarten function}~\cite{collins06}
\[ \operatorname{Wg}(\pi) := \frac{1}{(k!)^2} \sum_{\lambda \vdash k} \frac{(f^{\lambda})^2}{s_\lambda(1^{\times d})}\chi^{\lambda}(\pi). \]
This expression uses standard notation from the representation theory of the symmetric group (see~\cite{audenaert06} for an accessible introduction). The sum is over partitions $\lambda$ of $\{1,\dots,k\}$; $f^\lambda$ is the number of standard Young tableaux with shape $\lambda$; $\chi^\lambda(\pi)$ is the character of the symmetric group corresponding to partition $\lambda$, evaluated at $\pi$; and finally $s_{\lambda}(1^{\times d})$ is the Schur function corresponding to the partition $\lambda$, which has the explicit expression
\[ s_{\lambda}(1^{\times d}) = \frac{f^\lambda}{k!} \prod_{(i,j) \in \lambda} (d+j-i). \]
Then we have the following lemma, which is well-known but we prove for completeness in Appendix \ref{sec:inverse}.
\begin{lem}
\label{lem:ainv}
\[ A^{-1}_{\pi \sigma} = d^k \operatorname{Wg}(\pi^{-1}\sigma). \]
\end{lem}
For small $k$, one can use this expression to calculate the coefficients $\alpha_{\pi}$ exactly, but as $k$ increases the Weingarten function becomes cumbersome to work with. In order to address this, we now give a general upper bound on this function.

\begin{lem}
\label{lem:wgupper}
For any $k \le d^{2/3}$,
\[ |\operatorname{Wg}(\pi)| \le \frac{3\,C_{k-1}}{2}\,d^{c(\pi)-2k}, \]
where $C_{k-1}$ is the $(k-1)$'th Catalan number.
\end{lem}

The reader familiar with the work \cite{collins06} may wonder why we do not use the seemingly tighter bound given there that
\[ \operatorname{Wg}(\pi) = d^{c(\pi)-2k}(\operatorname{Moeb}(\pi) + O(d^{-2})), \]
where $\operatorname{Moeb}(\pi)$ is a function which can be shown to be upper bounded by $C_{k-1}$; the reason is that the $O(d^{-2})$ term in this bound hides an unspecified dependence\footnote{I would like to thank Carlos Gonz\'alez Guill\'en for pointing this out.} on $k$. Also note that the very recent work~\cite{collins12} gives an elegant alternative upper bound on the Weingarten function, but this bound does not seem to suffice for us to obtain tight results.

To prove Lemma \ref{lem:wgupper}, we will use a result of Matsumoto and Novak~\cite{novak10,matsumoto10,matsumoto12} which states that the Weingarten function can be expanded in terms of primitive factorisations in the symmetric group. This will allow us to use combinatorial bounds on such factorisations to obtain corresponding bounds on the Weingarten function. A primitive factorisation of a permutation $\pi \in S_k$ is a sequence of transpositions $(s_1,t_1),\dots,(s_\ell,t_\ell)$ such that the product $(s_1,t_1)\dots(s_\ell,t_\ell) = \pi$ and $t_1 \le t_2 \le \dots \le t_\ell$. Define $w_{\ell}(\pi)$ to be the number of primitive factorisations of $\pi$ into $\ell$ transpositions.

\begin{thm}[\cite{novak10,matsumoto10,matsumoto12}]
\label{thm:wgprimitive}
For any $k \le d$, and any $\pi \in S_k$,
\[ \operatorname{Wg}(\pi) = \frac{1}{d^k} \sum_{\ell = 0}^{\infty} w_{\ell}(\pi) \left(\frac{-1}{d} \right)^{\ell}. \]
\end{thm}

Note that this expansion implies that $w_{\ell}(\pi)$ depends only on the cycle type of $\pi$. Also observe that, by the same argument as used in Section \ref{sec:comb}, $w_{\ell}(\pi) = 0$ unless $\ell = k - c(\pi) + 2g$, for integer $g \ge 0$. In the case where $\pi$ is a cycle of length $k$ (called a $k$-cycle in what follows), $\operatorname{Wg}(\pi)$ can be evaluated explicitly.

\begin{thm}[\cite{collins03,matsumoto12}]
\label{thm:wgcycle}
Fix $k \le d$ and let $\kappa \in S_k$ be an arbitrary $k$-cycle. Then
\[ \operatorname{Wg}(\kappa) = \frac{(-1)^{k+1} C_{k-1}}{d(d^2-1^2) \dots (d^2 - (k-1)^2)}, \]
where $C_{k-1}$ is the $(k-1)$'th Catalan number.
\end{thm}

\begin{proof}[Proof of Lemma \ref{lem:wgupper}]
We first show that, for arbitrary $\pi \in S_k$, and arbitrary integer $g\ge 0$,
\be
\label{eq:wcyc} w_{k-c(\pi)+2g}(\pi) \le w_{k-1+2g}(\kappa),
\ee
where $\kappa$ is an arbitrary $k$-cycle. For any permutation $\pi$, there exists a $k$-cycle $\kappa_{\pi}$ such that $\kappa_{\pi}$ can be obtained from $\pi$ using $d(\kappa_{\pi},\pi)$ primitive transpositions (simply apply the sequence of transpositions $(j_1,k)(j_2,k)\dots(j_{c(\pi)-1},k)$ to $\pi$, where $\{j_1,\dots,j_{c(\pi)-1}\}$ contains one element from each of the cycles of $\pi$, except the cycle containing $k$). Thus each distinct primitive factorisation of $\pi$ of length $a$ gives a distinct primitive factorisation of $\kappa_{\pi}$ of length exactly $a + c(\pi) - 1$, which proves inequality (\ref{eq:wcyc}). Hence we have
\beas
|\operatorname{Wg}(\pi)| &=& \frac{1}{d^k} \sum_{\ell=1}^{\infty} w_{\ell}(\pi) d^{-\ell}
= d^{c(\pi)-2k} \sum_{g=0}^{\infty} w_{k-c(\pi)+2g} d^{-2g}
\le d^{c(\pi)-2k} \sum_{g=0}^{\infty} w_{k-1+2g}(\kappa) d^{-2g}\\
&=& d^{c(\pi)-1}|\operatorname{Wg}(\kappa)|,
\eeas
where the inequality is (\ref{eq:wcyc}). Now we can upper bound
\[ d^{2k-1}|\operatorname{Wg}(\kappa)| = \frac{C_{k-1}}{(1-1^2/d^2) \dots (1 - (k-1)^2/d^2)} \le \frac{C_{k-1}}{1 - \frac{1}{d^2} \sum_{i=1}^{k-1} i^2} \le \frac{3\,C_{k-1}}{2}, \]
where we use $k^3 \le d^2$. This completes the proof of Lemma \ref{lem:wgupper}.
\end{proof}

Lemma \ref{lem:wgupper} gives that
\be \label{eq:ainvbound} |A^{-1}_{\pi \sigma}| \le \frac{3\,C_{k-1}}{2}\, d^{c(\pi^{-1}\sigma)-k} \ee
for all $\pi,\sigma \in S_k$. We now use this to obtain an upper bound on the coefficients occurring in the decomposition of the specific matrix $M^{(k)}$ in terms of permutations, and hence prove Lemma \ref{lem:alphabound}. In fact, we prove the following lemma, from which Lemma \ref{lem:alphabound} immediately follows by observing that $C_{k-1} \le 4^{k-1}/k$.

\begin{lem}
Let $M_0$ be the projector onto a fixed $r$-dimensional subspace of $\C^d$, let $k$ be an arbitrary integer satisfying $1 \le k \le (r/2)^{2/3}$, and set
\[ M^{(k)} = \E_U[U^{\otimes k} M_0^{\otimes k} (U^{\dag})^{\otimes k}]. \]
Write $M^{(k)} = \sum_{\pi \in S_k} \alpha_\pi D_d(\pi)$. Then
\[ |\alpha_\pi| \le 3\,C_{k-1} k^2 4^{k-1} \frac{r^{c(\pi)}}{d^k}. \]
\end{lem}

\begin{proof}
As $\tr[M^{(k)} D_d(\pi^{-1})] = r^{c(\pi)}$ for all $\pi \in S_k$, by (\ref{eq:ainvbound}) we have
\[ \alpha_\pi = \frac{1}{d^k} \sum_{\sigma \in S_k} A^{-1}_{\pi \sigma} r^{c(\sigma)} \le \frac{3\,C_{k-1}}{2} \sum_{\sigma \in S_k} d^{c(\pi^{-1}\sigma)-2k} r^{c(\sigma)}. \]
Rewriting the sum and using Lemma \ref{lem:cycgenub}, we obtain
\beas |\alpha_\pi| &\le& \frac{3\,C_{k-1}}{2}\, d^{-2k} \sum_{i=1}^k d^i \sum_{\sigma,c(\pi^{-1}\sigma)=i} r^{c(\sigma)}\\
&=& \frac{3\,C_{k-1}}{2}\, d^{-2k} \sum_{i=1}^k d^i \sum_j r^j |\{\sigma:c(\pi^{-1}\sigma)=i,c(\sigma)=j\}|\\
&\le& \frac{3\,C_{k-1}}{2}\, d^{-2k} \sum_{i=1}^k d^i \sum_{j=1}^{k+c(\pi)-i} r^j 4^{k-1} k^{(3/2)(k + c(\pi) - (i+j) )+1}\\
&=& 3\,C_{k-1} d^{-2k} 2^{2k-3} k^{(3/2)(k + c(\pi))+1} \sum_{i=1}^k \left(\frac{d}{k^{3/2}}\right)^{\!i} \sum_{j=1}^{k+c(\pi)-i} \left(\frac{r}{k^{3/2}}\right)^j\\
&\le& 3\,C_{k-1} d^{-2k} 4^{k-1} k^{(3/2)(k + c(\pi))+1} \sum_{i=1}^k \left(\frac{d}{k^{3/2}}\right)^{\!i} \left(\frac{r}{k^{3/2}}\right)^{k+c(\pi)-i}\\
&=& 3\,C_{k-1} d^{-2k} k 4^{k-1} r^{k+c(\pi)} \sum_{i=1}^k \left(\frac{d}{r}\right)^i\\
&\le& 3\,C_{k-1} d^{-2k} k^2 4^{k-1} r^{k+c(\pi)} \left(\frac{d}{r}\right)^k\\
&=& 3\,C_{k-1} d^{-k} k^2 4^{k-1} r^{c(\pi)}.
\eeas
\end{proof}
Most of these steps are self-explanatory; in the third inequality we use the fact that $r/2 \ge k^{3/2}$, and in the fourth the fact that $r \le d$.


\section{Conclusion}
\label{sec:conclusion}

We have shown that random channels obey weak multiplicativity with high probability. When combined with the results of Christandl, Schuch and Winter~\cite{christandl09,christandl10} on the antisymmetric subspace, this implies that two of the constructions of channels which display the strongest known two-copy multiplicativity violations are in fact weakly multiplicative when the number of copies increases. This naturally leads one to conjecture that in fact {\em all} channels satisfy some form of weak multiplicativity (see~\cite{hastings09} for a similar conjecture). However, note that the proof strategy used here fails badly for the antisymmetric subspace, as
\[ P_{\operatorname{anti}}^{\Gamma} = \frac{1}{2} \left( I - F^\Gamma \right) = \frac{1}{2} \left( I - d\,\Phi \right), \]
where $F$ is the flip (swap) operator which interchanges two systems and $\Phi$ is a maximally entangled state. Hence we have $\|P_{\operatorname{anti}}^{\Gamma}\|_\infty = (d-1)/2$, which does not provide a useful upper bound on $h_{\operatorname{SEP}}(P_{\operatorname{anti}})$. It would be very interesting to find a single proof strategy which works in both cases.

It is also interesting to observe that our main result that random channels obey weak $p$-norm multiplicativity with exponent $(1/2-o(1))(1-1/p)$ becomes weaker as $p\rightarrow 1$, whereas known violations of multiplicativity are strongest at $p=\infty$. It is therefore tempting to conjecture that this exponent could be improved to $1/2-o(1)$ for all $p>1$, or even to $(1/2-o(1))f(p)$ for some function $f(p)$ which {\em decreases} with $p$.


\subsection*{Acknowledgements}

I would like to thank Graeme Mitchison for helpful discussions at an early stage of this work, and in particular for pointing out reference~\cite{audenaert06}. I would also like to thank Fernando Brand\~{a}o, Toby Cubitt, Carlos Gonz\'alez Guill\'en and Andreas Winter for helpful comments on a previous version, and Aram Harrow for sending me an early version of the work~\cite{harrow12}. Finally, I would like to thank two referees for their useful comments. This work was supported by an EPSRC Postdoctoral Research Fellowship.


\appendix

\section{Proof of Lemma \ref{lem:ainv}: Exact expression for \texorpdfstring{$A^{-1}$}{the inverse of A}}
\label{sec:inverse}

We finally prove Lemma \ref{lem:ainv}.

\begin{replem}{lem:ainv}
\[ A^{-1}_{\pi \sigma} = \frac{d^k}{(k!)^2} \sum_{\lambda \vdash k} \chi^{\lambda}(\pi^{-1} \sigma) \frac{(f^{\lambda})^2}{s_\lambda(1^{\times d})} = d^k \operatorname{Wg}(\pi^{-1} \sigma). \]
\end{replem}

\begin{proof}
We simply verify that
\beas
(A^{-1}A)_{\pi \sigma} &=& \sum_{\tau \in S_k} A^{-1}_{\pi \tau} A_{\tau \sigma}\\
&=& \frac{1}{(k!)^2} \sum_{\tau \in S_k} \left(\sum_{\lambda \vdash k} \chi^{\lambda}(\pi^{-1} \tau) \frac{(f^{\lambda})^2}{s_\lambda(1^{\times d})} \right) d^{c(\tau^{-1}\sigma)}\\
&=& \frac{1}{(k!)^2} \sum_{\lambda,\mu \vdash k} \frac{(f^{\lambda})^2}{s_\lambda(1^{\times d})} s_\mu(1^{\times d}) \sum_{\tau \in S_k} \chi^{\lambda}(\pi^{-1} \tau) \chi^{\mu}(\tau^{-1} \sigma) \\
&=& \frac{1}{k!} \sum_{\lambda,\mu \vdash k} \frac{f^{\lambda}}{s_\lambda(1^{\times d})} s_\mu(1^{\times d}) \chi^{\lambda}(\pi^{-1}\sigma) \delta_{\lambda \mu}\\
&=& \frac{1}{k!} \sum_{\lambda \vdash k} f^{\lambda} \chi^{\lambda}(\pi^{-1}\sigma) = \delta_{\pi \sigma}.
\eeas
The third equality above is the well-known relation $\sum_{\mu \vdash k} s_\mu(1^{\times d}) \chi^{\mu}(\pi) = d^{c(\pi)}$~\cite{audenaert06}, and the fourth and sixth are character orthogonality relations (observing that $f^\lambda = \chi^\lambda(e)$).
\end{proof}


\bibliographystyle{plain}
\bibliography{../thesis}

\end{document}